\documentclass[aps,pra,twocolumn,superscriptaddress,showpacs,floatfix,nofootinbib,llncs,longbibliography,pdftex]{revtex4-2}
\usepackage{graphicx}
\usepackage{amssymb}
\usepackage{amsmath}
\usepackage{color}
\usepackage{psfrag}
\usepackage{epsfig}
\usepackage{bbm}
\usepackage{bm}
\usepackage{hyperref}
\usepackage[normalem]{ulem}
\usepackage{amsthm}
\usepackage{epstopdf}
\usepackage{verbatim}
\usepackage{subfigure}
\usepackage[export]{adjustbox}
\usepackage{bbold}
\newtheorem{theorem}{Theorem}
\usepackage{array}
\usepackage{mathtools}

\newcommand{\ket}[1]{| #1 \rangle}
\newcommand{\bra}[1]{\langle #1 |}

\newcommand{\tr}{\text{Tr}}

\newcommand{\s}{\text{S}}

\newcommand{\e}{\text{E}}

\newcommand{\Ea}{\text{E}_\text{A}}

\newcolumntype{M}[1]{>{\centering\arraybackslash}m{#1}}

\begin{document}
\title{Resource-Efficient Certification of System–Environment Entanglement Solely from Reduced System Dynamics}
\author{Jhen-Dong Lin}
\email{jhendonglin@gmail.com}
\affiliation{Department of Physics, National Cheng Kung University, 701 Tainan, Taiwan}
\affiliation{Center for Quantum Frontiers of Research \& Technology, NCKU, 70101 Tainan, Taiwan}

\author{Pao-Wen Tu}
\affiliation{Department of Physics, National Cheng Kung University, 701 Tainan, Taiwan}
\affiliation{Center for Quantum Frontiers of Research \& Technology, NCKU, 70101 Tainan, Taiwan}

\author{Kuan-Yi Lee}
\affiliation{Department of Physics, National Cheng Kung University, 701 Tainan, Taiwan}
\affiliation{Center for Quantum Frontiers of Research \& Technology, NCKU, 70101 Tainan, Taiwan}

\author{Neill Lambert}
\affiliation{RIKEN Center for Quantum Computing, RIKEN, Wakoshi, Saitama 351-0198, Japan}

\author{Adam Miranowicz}
\affiliation{RIKEN Center for Quantum Computing, RIKEN, Wakoshi, Saitama 351-0198, Japan}
\affiliation{Institute of Spintronics and Quantum Information, Faculty of Physics and Astronomy, Adam Mickiewicz University, 61-614 Pozna\'n, Poland}

\author{Franco Nori}
\affiliation{RIKEN Center for Quantum Computing, RIKEN, Wakoshi, Saitama 351-0198, Japan}
\affiliation{Physics Department, The University of Michigan, Ann Arbor, Michigan 48109-1040, USA.}

\author{Yueh-Nan Chen}
\email{yuehnan@mail.ncku.edu.tw}
\affiliation{Department of Physics, National Cheng Kung University, 701 Tainan, Taiwan}
\affiliation{Center for Quantum Frontiers of Research \& Technology, NCKU, 70101 Tainan, Taiwan}
\affiliation{Physics Division, National Center for Theoretical Sciences, Taipei 106319, Taiwan}

\date{\today}

\begin{abstract}
Certifying nonclassical correlations typically requires access to all subsystems, presenting a major challenge in open quantum systems coupled to inaccessible environments. Recent works have shown that, in autonomous pure dephasing scenarios, quantum discord with the environment can be certified from system-only dynamics via the Hamiltonian ensemble formulation. However, this approach leaves open whether stronger correlations, such as entanglement, can be certified. Moreover, its reliance on Fourier analysis requires full-time dynamics, which is experimentally resource-intensive and provides limited information about when such correlations are established during evolution. In this work, we present a method that enables the certification of system-environment quantum entanglement solely from the reduced dynamics of the system. The method is based on the theory of mixed-unitary channels and applies to general non-autonomous pure dephasing scenarios. Crucially, it relaxes the need for full-time dynamics, offering a resource-efficient approach that also reveals the precise timing of entanglement generation. We experimentally validate this method on a Quantinuum trapped-ion quantum processor with a controlled-dephasing model. Finally, we highlight its potential as a tool for certifying gravitationally induced entanglement.
\end{abstract}

\maketitle

\section{Introduction}

Quantum entanglement is a distinct form of quantum correlation with no classical counterpart~\cite{RevModPhys.81.865}. This nonclassical resource plays a central role in quantum computation~\cite{nielsen2010quantum}, communication~\cite{gisin2007quantum}, metrology~\cite{PhysRevLett.96.010401, RevModPhys.89.035002}, cryptography~\cite{
RevModPhys.74.145, pirandola2020advances}, thermodynamics~\cite{RevModPhys.81.1, parrondo2015thermodynamics,kurizki2022thermodynamics}, and in understanding critical phenomena in complex many-body systems~\cite{ RevModPhys.80.517, RevModPhys.91.021001}.

However, certifying entanglement remains a long-standing and generally challenging task. Full access to all subsystems is typically required, which becomes infeasible, especially in open quantum systems where the system of interest interacts with inaccessible environments~\cite{isar1994open, breuer2002theory, rivas2012open, lidar2019lecture, minganti2024open}. Building upon the framework of open quantum systems, several approaches have been developed to partially infer (nonclassical) system–environment correlations, either through the knowledge of system–environment interactions~\cite{shi2009electron,zhu2012explicit,PhysRevA.84.062121,iles2014environmental,PhysRevB.107.085428,PhysRevResearch.6.033237,PhysRevB.109.115408,PhysRevA.112.012211}, or via local detection schemes~\cite{PhysRevLett.107.180402,smirne2011experimental, mazzola2012dynamical, smirne2013interaction, gessner2014local,
tang2015experimental}, which exploit quantum non-Markovian memory effects~\cite{rivas2014quantum,breuer2016colloquium,de2017dynamics}.

One particularly promising class of open-system scenarios is pure dephasing~\cite{RevModPhys.75.715, schlosshauer2019quantum}, where the system loses quantum coherence while its level populations remain approximately unchanged. Pure dephasing arises in many physical settings, such as impurity-phonon interactions in solid-state systems~\cite{PhysRevB.65.195313}, circuit and cavity quantum electrodynamics in the dispersive regime~\cite{walther2006cavity, RevModPhys.93.025005, gu2017microwave, kockum2019quantum, kjaergaard2020superconducting}, and even in certain theoretical models of quantum gravitational interactions between massive particles~\cite{ PhysRevLett.119.240401, PRXQuantum.2.030330, PhysRevLett.128.110401, PRXQuantum.5.010318, RevModPhys.97.015006, RevModPhys.97.015003}. 

With negligible population changes, pure dephasing dynamics offer a simplified yet non-trivial arena for studying system-environment correlations. A notable recent advance builds upon the Hamiltonian ensemble (HE) representation~\cite{PhysRevLett.120.030403, chen2019quantifying, chen2024unveiling}, where the non-unitary system dynamics is described as an ensemble average over autonomous Hamiltonian evolutions. Importantly, not all types of pure dephasing dynamics are compatible with this representation. Such HE incompatibility, under the assumption of autonomous evolution, reveals the presence of system-environment quantum discord~\cite{PhysRevLett.88.017901, henderson2001classical, RevModPhys.84.1655, bera2017quantum}, a weaker form of nonclassical correlation than quantum entanglement.

Nevertheless, the HE incompatibility approach comes with certain challenges. First, as aforementioned, the inference of quantum discord depends on the assumption of autonomous global evolution; that is, evolution under a time-independent Hamiltonian. Second, detecting HE incompatibility requires Fourier analysis of the reduced dynamics over the entire time domain, which is experimentally demanding and provides no information about the specific time at which quantum discord is established. Finally,  whether this method can be extended to witness stronger correlations, such as quantum entanglement, remains unexplored.

In this work, we introduce a generalized entanglement detection scheme that overcomes these limitations. To this end, we employ the concept of mixed-unitary (MU) channels~\cite{audenaert2008random, PhysRevA.80.042108, mendl2009unital, PhysRevA.92.052117, PhysRevA.93.032132, lee2020detecting}. These channels can be expressed as probabilistic mixtures of unitary evolutions and naturally generalize the HE representation beyond strictly autonomous dynamics. Using recently proposed system–environment entanglement criteria for pure dephasing~\cite{PhysRevA.92.032310,PhysRevA.98.052344,PhysRevResearch.2.043062,zhan2021experimental,PhysRevA.104.042411,takahashi2025coherence}, we show that \textit{non-mixed unitarity (NMU) of the reduced dynamics serves as an experimentally efficient witness of system–environment entanglement.} This approach retains the operational spirit of the HE incompatibility while significantly broadening its scope:

\begin{enumerate}
\item The non-mixed unitarity method applies to non-autonomous dynamics, relaxing the requirement of strictly autonomous evolution.
\item It enables certification of system–environment entanglement at arbitrary times, reducing the need for full-time dynamics and lowering the experimental demands.
\item The method remains effective even when the dynamics exhibit pure dephasing behavior only over a short-time regime, unlike the HE model, which requires strict pure dephasing throughout the entire evolution.
\end{enumerate}

We validate this approach by implementing controlled-dephasing dynamics on a trapped-ion quantum processor, System Model H1 by Quantinuum~\cite{quantinuum}. The results demonstrate that certifying HE incompatibility requires sampling the system’s dynamics at multiple time points. In contrast, system-environment entanglement can be certified from a single time point whenever the corresponding NMU is non-trivial, demonstrating its practical advantage in experiments.

Furthermore, we propose that the NMU provides a promising approach for certifying gravitationally induced entanglement in scenarios where massive particles interact with each other via a Newtonian potential. In contrast to the coherence-revival detection scheme introduced in Ref.~\cite{PRXQuantum.2.030330}, the NMU method does not rely on the assumption of time-translational invariance in the system’s dynamics. Importantly, the NMU method enables an unambiguous certification of entanglement generated by gravity, rather than merely excluding specific classes of semiclassical models, as discussed in the corresponding Erratum to Ref.~\cite{PRXQuantum.2.030330}.

The remainder of this paper is organized as follows: In Sec.~\ref{sec2}, we introduce the theoretical framework and show that NMU serves as a witness of system-environment entanglement and compare it with HE incompatibility. In Sec.~\ref{sec3}, we detail the experimental realization on the Quantinuum trapped-ion processor. In Sec.~\ref{sec4}, we discuss the application of the NMU approach to certifying gravitationally induced entanglement. Finally, in Sec.~\ref{sec5}, we summarize this work and discuss future directions.

\section{Certify system-environment entanglement of pure dephasing \label{sec2}}

In this work, we consider pure dephasing of a generic open system (S) with respect to a reference basis $\{|i\rangle\}$. The dephasing arises from the interaction between the system and its environment (E), governed by a global system-environment unitary evolution of the form:
\begin{equation}
    U_{\s\e} = \sum_i \ket{i}\bra{i} \otimes V_i, \label{eq: U_SE}
\end{equation}
where $\{V_i\}$ are unitary operators acting on E. Given an initially uncorrelated state $\rho_\s\otimes \rho_\e$, the joint system-environment state after evolution becomes 
\begin{equation}
    \begin{aligned}
        \rho_{\s\e} = U_{\s\e}~ \rho_\s \otimes \rho_\e~ U_{\s\e}^\dagger = \sum_{i,j} \rho_{\s,i,j} \ket{i}\bra{j} \otimes V_i \rho_\e V_j^\dagger
    \end{aligned}
\end{equation}
with $\rho_{\s,i,j} = \bra{i}\rho_\s \ket{j}$. The evolved system's reduced state can then be written as 
\begin{equation}
    \rho'_\s  = \tr_\e(\rho_{\s\e})= \sum_{i,j} \rho_{\s,i,j}~ \phi_{i,j} \ket{i}\bra{j},
\end{equation}
with the dephasing factor $\phi_{i,j} = \tr_\e (V_i \rho_\e V_j^\dagger) $. Since $\phi_{i,i} = 1$ for all $i$, the populations remain unchanged, reflecting the pure dephasing nature of the system's evolution.

If the system is initially prepared with vanishing coherence (i.e. $\rho_{\s,i,j} = 0~\forall~ i\neq j$), the evolved joint state $\rho_{\s\e}$ remains separable and takes the form $\rho_{\s\e} = \sum_i \rho_{\s,i,i} \ket{i}\bra{i} \otimes V_i \rho_\e V_i^\dagger$. This indicates that the presence of initial coherence in S is a necessary condition for establishing system-environment entanglement~\cite{takahashi2025coherence}. However, it is not sufficient: the structure of the global unitary $U_{\s\e}$ and the initial environmental state $\rho_\e$ also plays a crucial role in the entanglement generation. For instance, consider the case where $V_i$ and $\rho_\e$ can be simultaneously diagonalized in some environmental basis $\{\ket{e_n}\}$, such that $V_i = \sum_n \exp(i\phi_{i,n})\ket{e_n}\bra{e_n}$ and $\rho_\e = \sum_n p_n|e_n\rangle \langle e_n|$. Under this condition, the evolved total state becomes
\begin{equation}
    \begin{aligned}
        &\tilde\rho_{\s\e} = \sum_n p_n~\rho_{\s,n} \otimes   \ket{e_n}\bra{e_n}, \\
        \text{with } &\rho_{\s,n} =  \sum_{i,j} ~ \rho_{\s,i,j}~ e^{i(\phi_{i,n}-\phi_{j,n})}~ \ket{i}\bra{j}.
    \end{aligned}
\end{equation}
Hence, even in the presence of initial coherence, $\tilde \rho_{\s\e}$ remains separable. More precisely, it is also zero discordant from E to S. Moreover, the reduced dynamics of the system in this case correspond to a MU channel:
\begin{equation}
    \tilde\rho_\s = \tr_\e \tilde\rho_{\s\e} = \sum_n p_n \tilde{U}_n \rho_\s \tilde{U}_n^\dagger,
\end{equation}
where each $\tilde{U}_n = \sum_i \exp(i\phi_{i,n})\ket{i}\bra{i}$ is a unitary operator acting on S. 

In this particular case, we observe that the separability of the evolved joint state leads to a MU channel for the reduced dynamics. Although this observation arises from a specific scenario, it raises a broader question: can deviations from MU dynamics serve as signatures of system-environment entanglement? As the main result of this work, we formalize this possibility through the following theorems.

\begin{theorem}
The reduced evolution of a system is a mixed-unitary channel if the evolved system-environment state $\rho_{\s\e}$ remains zero-discordant  from E to S.\label{thm1}
\end{theorem}
\begin{proof}
The proof begins by assuming that $\rho_{\s\e}$ is zero-discordant from E to S. This implies that there exist a rank-one projective measurement $\{M_n\}$ on E such that $\sum_n\mathbb{1}\otimes M_n~ \rho_{\s\e}~ \mathbb{1}\otimes M_n =\rho_{\s\e}$. As shown in Ref.~\cite{huang2011new}, this condition can be expressed as 
\begin{equation}
    [V_i\rho_\e V_j^\dagger, V_k \rho_\e V_l^\dagger] = 0~~~~\forall~i,j,k,l. \label{eq: zero discord criteria}
\end{equation}
By choosing $k=j\text{ and } l=i$, we obtain 
\begin{equation}
    \begin{aligned}
        &V_i~\rho_\e^2~V_i^\dagger = V_j~\rho_\e^2~V_j^\dagger, \\
        \text{or equivalently, } &W_{ij}~\rho_\e^2~W_{ij}^\dagger = \rho_\e^2, \label{eq: zero discord criteria 2}
    \end{aligned}
\end{equation}
where we define the unitary operator $W_{ij} = V_j^\dagger V_i$.
Let us now consider the spectral resolution: $\rho_\e = \sum_n p_n \ket{e_n}\bra{e_n}$. For simplicity, we assume here that $\rho_\e$ is full-rank and non-degenerate, which implies that $p_n>0~\forall n$, $p_n\neq p_m~\forall n\neq m$, and $\{\ket{e_n}\}$ forms a basis of E. Under this assumption, we conclude that
\begin{equation}
      W_{ij} \ket{e_n}\bra{e_n} W_{ij}^\dagger = \ket{e_n}\bra{e_n} ~~~~\forall n.
\end{equation} 
Accordingly, the unitary operator $W_{ij}$ is diagonal in the basis $\ket{e_n}$ and can be expressed as 
\begin{equation}
    W_{ij} = \sum_n\exp(i\theta_{i,j,n})\ket{e_n}\bra{e_n}.
\end{equation}
Furthermore, one can observe that $W_{ij}W_{ik}^\dagger = W_{kj}$, which implies that $\theta_{i,j,n}-\theta_{i,k,n} = \theta_{k,j,n}$. In other words, the phase $\theta_{i,j,n}$ can be written as a difference of two other phases $\phi_{i,n}$ and $\phi_{j,n}$, i.e., $\theta_{i,j,n} = \phi_{i,n}-\phi_{j,n}$. This allows to conclude that the reduced evolution of S corresponds to a MU channel:
\begin{equation}
    \begin{aligned}
        \tilde{\rho}_\s = \sum_{i,j} \rho_{\s,i,j} \tr(W_{ij}\rho_\e) \ket{i}\bra{j} = \sum_n p_n \tilde{U}_n \rho_\s \tilde{U}_n^\dagger,
    \end{aligned}
\end{equation}
with $\tilde{U}_n = \sum_i \exp(i\phi_{i,n})\ket{i}\bra{i}$.
Note that this derivation can be extended beyond the assumption that $\rho_\e$ is full-rank and non-degenerate. We leave this formal proof in Appendix~\ref{Appendix A}.  
\end{proof}
In general, it is known that if a bipartite state is zero discordant, it must also be separable; however, the converse does not necessarily hold. Nevertheless, as we demonstrate in Theorem~\ref{thm2} below, when the system is prepared in a pure state, the separability of the total state $\rho_{\s\e}$ also implies that it is zero discordant from E to S. 

\begin{theorem}
The criteria for $\rho_{\s\e}$ to remain zero discordant and separable are equivalent. \label{thm2}
\end{theorem}
\begin{proof}
According to Refs.~\cite{PhysRevA.92.032310, PhysRevA.98.052344}, the separability criteria for the pure dephasing model, with initial pure system state, can be expressed through the following two independent classes of conditions:
\begin{equation}
    V_i \rho_\e V_i^\dagger = V_j \rho_\e V_j^\dagger=R~~~~\forall i,j, \label{eq: qubit-like condition}
\end{equation}
and
\begin{equation}
    [V_iV_j^\dagger,V_k V_l^\dagger]=0~~~~\forall i,j,k,l. \label{eq: qutrit-like condition}
\end{equation}
These two classes of conditions are referred to qubit-like and qutrit-like conditions, since they can be obtained by treating the system as a qubit and a qutrit, respectively.
Combining these two conditions, we show that
\begin{equation}
    \begin{aligned}
        &R[V_iV_j^\dagger,V_k V_l^\dagger]R= [V_i\rho_\e V_j^\dagger,V_k \rho_\e V_l^\dagger]=0.
    \end{aligned}
\end{equation}
This leads to the conclusion that the separability also implies zero discord. 

We remark that the qutrit-like conditions described by Eq.~\eqref{eq: qutrit-like condition}, which was originally derived in Ref.~\cite{PhysRevA.98.052344}, are overly restrictive. The correct form is derived in Appendix~\ref{Appendix B}. The formal proof of Theorem~\ref{thm2}, incorporating this correction, is provided in Appendix~\ref{Appendix C}. 
\end{proof}

We can now formalize the main result of this work: \textit{Non-mixed unitarity for pure dephasing dynamics can serve as a witness for system-environment entanglement.} More precisely, if the reduced evolution of the system is not a MU channel, then there exist initial system states for which the evolved system-environment joint state becomes entangled. 

According to Refs.~\cite{PhysRevA.92.032310, PhysRevA.98.052344}, a violation of the separability criterion implies system-environment entanglement with nonzero negativity~\cite{PhysRevLett.77.1413, HORODECKI19961, PhysRevA.65.032314}. Based on this insight, we conjecture that NMU may be quantitatively related to the negativity. In Appendix~\ref{Appendix D}, we present numerical simulations supporting this conjecture, showing that NMU provides a lower bound on the negativity.

\subsection{Comparison with Hamiltonain ensemble formalism}
The proposed method based on NMU serves as a natural extension of the HE incompatibility originally introduced in Ref.~\cite{PhysRevLett.120.030403}. We briefly review the central concept here for comparison. A HE is defined as $\{(p_\lambda,H_\lambda)\}_\lambda$, which consists of a collection of time-independent Hamiltonian $H_\lambda$ acting on the system, each associated with a time-independent probability distribution $p_\lambda$. The corresponding ensemble-averaged system dynamics is described by the dynamical map $\{\mathcal{E}_t\}_t$:
\begin{equation}
    \rho_\s(t) = \mathcal{E}_t(\rho_\s) =  \int d\lambda~p_\lambda ~e^{-iH_\lambda t} ~\rho_\s~ e^{iH_\lambda t}. \label{eq: ensemble average dynamics}
\end{equation}
Operationally, a unique (quasi-)distribution $\tilde{p}_\lambda$ for a given pure dephasing dynamics can be obtained by performing a generalized inverse Fourier transform, as outlined in Ref.~\cite{PhysRevLett.120.030403}. Further, a given pure dephasing dynamics is incompatible with HE model, if the constructed quasi-distribution $\tilde{p}_\lambda$ exhibit negative values. Notably, under the assumption of autonomous pure dephasing, where the system-environment global evolution is governed by a time-independent Hamiltonian, \textit{the HE incompatibility also implies the generation of system-environment quantum discord}. 

Compared to the HE incompatibility, the proposed NMU approach provides several advantages for entanglement certification. First, according to Theorems~\ref{thm1} and \ref{thm2}, the NMU approach is not restricted to autonomous system-environment models. They are also applicable to more general non-autonomous pure dephasing settings, such as Floquet~\cite{mori2023floquet} and collisional models~\cite{ciccarello2022quantum}. 

Second, constructing the quasi-probability under the HE framework requires access to the full-time dynamics, which can result in substantial experimental overhead. Moreover, while the negativity of the quasi-probability indicates the presence of system-environment quantum discord, it offers no information about the specific time at which this nonclassical correlation is established. In contrast, the NMU framework enables certification of system-environment entanglement at any given time, thereby relaxing the stringent requirement of full-time dynamics imposed by the HE approach and significantly reducing the experimental demands. 

Third, the NMU also applies to the cases where the dynamics can be well approximated by pure dephasing only in the short-time regime. Specifically, for such a scenario, the short-time system-environment evolution can be written as $U(\tau) \approx \sum_i \ket{i}\bra{i}\otimes V_i(\tau) + \delta U(\tau)$, where $\delta U(\tau)$ captures the small deviation from pure dephasing with $\|\delta U(\tau)\|\ll 1$. In this case, entanglement generated at time $\tau$ is predominantly due to the pure dephasing term, which can be effectively certified using the NMU approach.

\section{Experiment on cloud quantum computer \label{sec3}}
\begin{figure}
    \centering
    \includegraphics[width=1\linewidth]{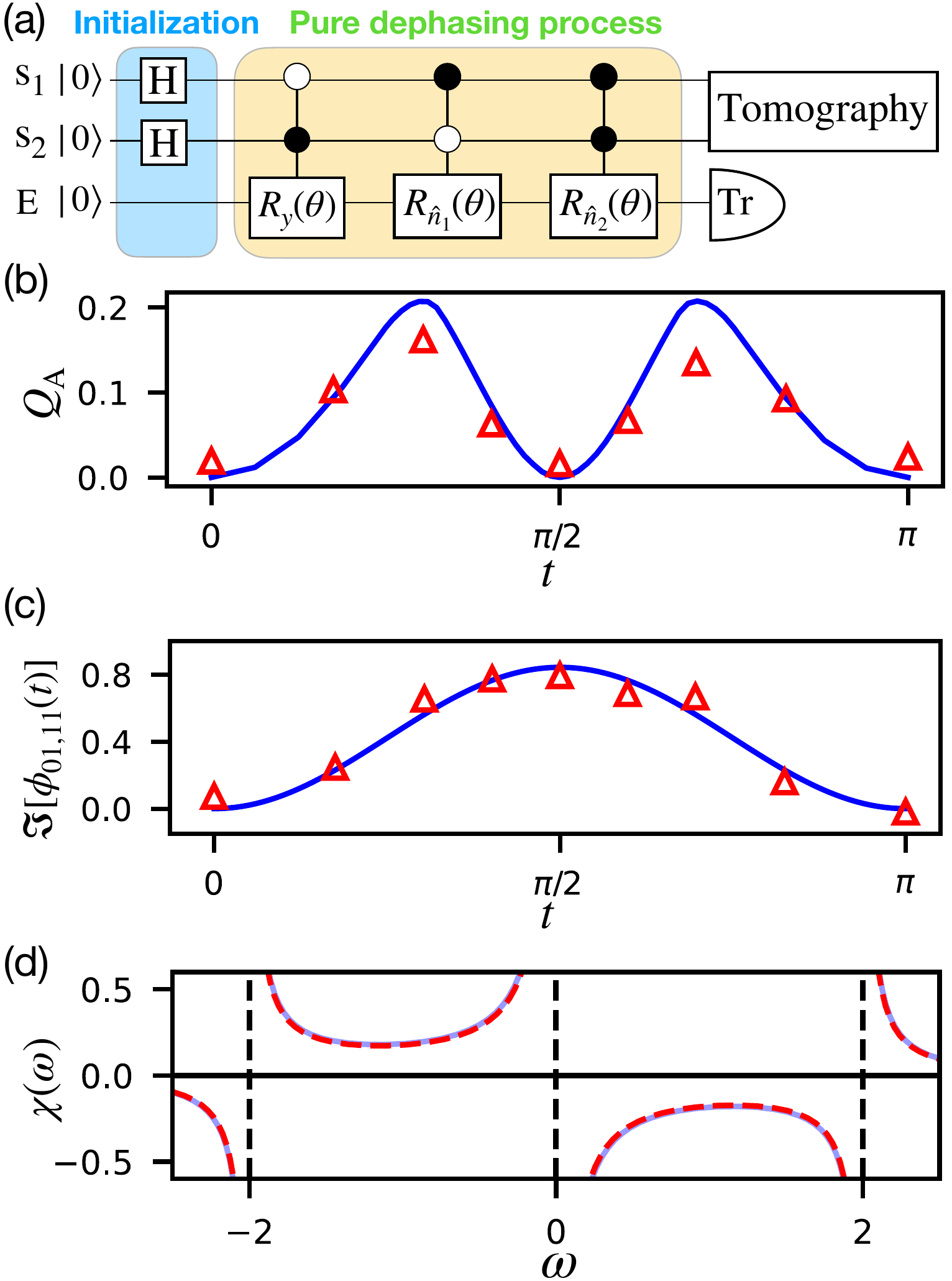}
    \caption{
    (a) Circuit implementation of a pure dephasing dynamics governed by Eq.~\eqref{eq: Hamiltoninan with single qubit environment}, where we set $\theta =2t$.
    (b) Theoretical (blue curve) and experimental (red triangles) results for the NMU measure $Q_A$ as a function of time.
    (c) The fitting curve (blue curve) and the experimental data (red triangles) for the dephasing factor $\Im[\phi_{01,11}(t)]$.
    (d) Comparison of the theoretical distribution $\chi(\omega)$ (solid blue curve) and the reconstructed distribution $\chi_\text{fit}(\omega)$ (red dashed curve).}
    \label{fig1}
\end{figure}
Here, we demonstrate the aforementioned advantages by applying it to a dephasing model implemented on the trapped-ion quantum processor, System Model H1 by Quantinuum~\cite{quantinuum}. The model consists of a two-qubit system interacting with a single-qubit environment, governed by the Hamiltonian
\begin{equation}
    H = \ket{01}\bra{01} \otimes \sigma_{y} + \ket{10}\bra{10} \otimes \hat{n}_{1} \cdot \vec{\sigma} + \ket{11}\bra{11} \otimes \hat{n}_{2} \cdot \vec{\sigma},\label{eq: Hamiltoninan with single qubit environment}
\end{equation}
with the unit vectors $\hat{n}_{1}=(\sqrt{3}/2,-1/2)$, $\hat{n}_{2}=(-\sqrt{3}/2,-1/2)$, and the two-dimensional Pauli vector $\vec\sigma = (\sigma_x, \sigma_y)$. The environment qubit is prepared in $\rho_\e = \ket{0}\bra{0}$. 

The circuit implementation of this model is presented in Fig.~\ref{fig1}(a) with the system qubits $\s_1$ and $\s_2$ and the environmental qubit E. The pure dephasing process is implemented by a sequence of three controlled-controlled rotation gates, where $R_{\hat{n}}(\theta)$ denotes a single-qubit rotation about the axis $\hat{n}$ by an angle $\theta$. The time evolution is simulated by tuning the rotation angle such that $\theta =2 t$. As presented in Appendix~\ref{Appendix E}, each controlled-controlled rotation gate can be decomposed into four two-qubit controlled-Z gates and several single-qubit gates, all of which are native to the Quantinuum device. Consequently, the entire circuit involves a total of twelve two-qubit gates.

To experimentally obtain the dynamics of pure dephasing factors, we prepare the system qubits in the maximally coherent state $(\ket{00}+\ket{01}+\ket{10}+\ket{11})/2$ using Hadamard (H) gates. State tomography is then performed to monitor the evolution of the off-diagonal elements, with the corresponding dephasing factors defined as $\rho_{\s,ij,kl}(t) = \phi_{ij,kl}(t)/4$. The environment qubit E is left unmeasured (i.e., traced out). The experimental results presented below are obtained from state tomography with 300 shots. 

Following Ref.~\cite{audenaert2008random}, the NMU of a channel $\Lambda$ can be quantified by the entanglement of assistance of its Choi state:
\begin{equation}
    \mathcal{J}(\Lambda) = \mathcal{I}_{\s'} \otimes \Lambda (\ket{\Phi^+}\bra{\Phi^+}).
\end{equation}
Here, $\mathcal{I}_{\s'}$ denotes the identity map on the ancillary system $\s'$, and $\ket{\Phi^+}$ is the maximally entangled state,
\begin{equation}
    \ket{\Phi^+} = \frac{1}{\sqrt{d_\s}}\sum_{i=0}^{d_\s-1}\ket{i_{\s'}}\otimes \ket{i_\s}.
\end{equation}
The entanglement of assistance is defined as~\cite{divincenzo1998entanglement}
\begin{equation}
\begin{aligned}
    \Ea[\mathcal{J}(\Lambda)] = \max_{\{p_i,\psi_i\}} \sum_i p_i S(\tr_\s \ket{\psi_i}\bra{\psi_i}),
\end{aligned}
\end{equation}
where the maximization runs over all pure-state decompositions, i.e. $\mathcal{J}(\Lambda) = \sum_i p_i \ket{\psi_i}\bra{\psi_i}$, and $S(\cdot)$ denotes the von-Neumann entropy.

For a MU channel, $\Lambda_{\text{MU}}(\bullet) = \sum_j p_j U_j \bullet U_j^\dagger $, the Choi state is a mixture of maximally entangled states: 
\begin{equation}
\begin{aligned}
    \mathcal{J}(\Lambda_{\text{MU}}) &= \sum_j p_j \ket{U_j}\bra{U_j},\\
    \ket{U_j} &= \frac{1}{\sqrt{d_\s}}\sum_{i=0}^{d_\s-1}\ket{i_{\s'}} \otimes U_j\ket{i_\s},
\end{aligned}
\end{equation}
which attains the maximal entanglement of assistance~\cite{audenaert2008random}:
\begin{equation}
    \Ea[\mathcal{J}(\Lambda_\text{MU})] = \text{E}_{A,\text{max}} =\log d_\s.
\end{equation}
Therefore, deviations from this maximum indicate NMU and can be quantified by
\begin{equation}
    Q_\text{A}(\rho_\Lambda) = \text{E}_\text{A,max} - \Ea[\mathcal{J}(\Lambda)], \label{eq: QA}
\end{equation} 
with $Q_\text{A}>0$ signaling the presence of NMU. 

As illustrated in Fig.~\ref{fig1}(b), the dynamics of $Q_\text{A}$ shows a clear signature of NMU. Both the theoretical predictions and experimental results agree well, exhibiting nonzero NMU except at the special times $t \neq 0,\pi/2,\pi$, where the channel reduces to a (nearly) MU channel. This confirms that deviations from the maximal entanglement of assistance provide a faithful indicator of NMU.

In addition to NMU, the model enables direct observation of HE incompatibility through the quasi-probability distribution $p(\omega)$ associated with the dephasing factor $\phi_{01,11}(t)$. The analytical expressions are
\begin{equation}
    \begin{aligned}
        \phi_{01,11}(t)&= \cos^2(t)-\frac{1}{2}\sin^2(t)+i \frac{\sqrt{3}}{2}\sin^2(t)\mathrm{sgn}(t),\\
        p(\omega) &= \frac{1}{2\pi} \int dt~\phi_{01,11}(t) e^{i\omega t} \\&=\frac{1}{4} \delta (\omega ) + \frac{3}{8} \delta (\omega -2)+ \frac{1}{8} \delta (\omega +2)+\chi(\omega). \label{eq: phi_01,11}
    \end{aligned}
\end{equation}
Following the procedure in Ref.~\cite{PhysRevLett.120.030403}, we manually introduce the factor $\mathrm{sgn}(t)$ in the imaginary part of the dephasing factor, $\Im[\phi_{01,11}(t)]$. This ensures that $p(\omega)$ is a real function, without altering the observed dephasing factor for $t\geq 0.$
The negativity arises through the final term 
\begin{equation}
    \chi(\omega) = \frac{\sqrt{3}}{\pi{\omega(\omega-2)(\omega+2)}},
\end{equation}
which is obtained from the Fourier transform of $\Im[\phi_{01,11}(t)]$. Specifically, $\chi(\omega)$ takes negative values when $\omega \in (-\infty, -2) \cup (0,2)$. 

To assess this feature experimentally, we fit the dephasing factor $\Im[\phi^{\text{fit}}_{01,11}(t)]$ to the experimental data and reconstruct the distribution $\chi_\text{fit}(\omega)$. Guided by the analytical expression in Eq.~\eqref{eq: phi_01,11}, we adopt the following ansatz:
\begin{equation}
    \Im[\phi^{\text{fit}}_{01,11}(t)] = \alpha \frac{\sqrt{3}}{2} \sin^2(t),
\end{equation}
with $\alpha$ as a fitting parameter, resulting in $\chi_\text{fit}(\omega) = \alpha\chi(\omega)$. The optimal fit, $\alpha\approx 0.9739$, shows good agreement with the experimental data [Fig.~\ref{fig1}(c)]. Moreover, as shown in Fig.~\ref{fig1}(d), the reconstructed distribution agrees well to the theoretical prediction and clearly exhibits negativity for $\omega<-2$ and $0<\omega<2$. This provides direct evidence of the HE incompatibility. Thus, system-environment entanglement is reflected in both the observed NMU and the negativity in the quasi-probability distribution.

Note that the model exhibits a built-in discrete time-translational symmetry: $\rho_{\s\e}(t+n\pi) = \rho_{\s\e}(t)$ with $n=1,2,3,\cdots$, which reduces the need to track the full-time dynamics to access HE incompatibility. Nevertheless, reconstructing $\chi_\text{fit}(\omega)$ still requires sampling the system's dynamics at multiple time points within $t\in[0,\pi]$. As shown in Fig.~\ref{fig1}(b), we use seven time points and obtain the best-fit dephasing factor to reconstruct the quasi-distribution. In contrast, NMU can be certified from a single time point as long as its value is positive. This illustrates that NMU offers a more economical approach to entanglement certification than HE incompatibility.

In addition, while a single dephasing factor is sufficient to reveal quasi-distribution negativity in this model, this feature is not generic to all pure dephasing processes with HE incompatibility. In general, dephasing factors correspond only to marginals of a higher-dimensional joint distribution that fully characterizes the underlying dynamics~\cite{chen2019quantifying}. Recently, machine learning techniques have been explored to aid in reconstructing these distributions~\cite{chen2024unveiling}; however, a general and systematic approach remains unavailable, and the task is computationally demanding even for two-qubit systems. By contrast, detecting system-environment entanglement via NMU offers a significantly more efficient alternative, both experimentally and computationally.

\section{Witnessing gravitationally induced entanglement \label{sec4}}
\begin{figure}
    \centering
    \includegraphics[width=0.95\linewidth]{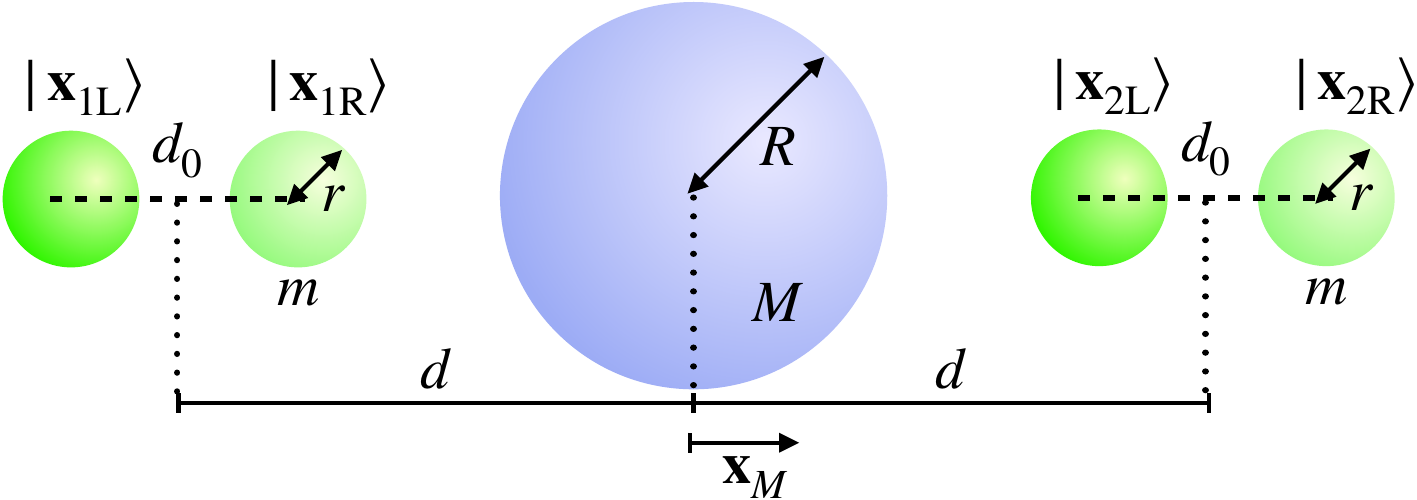}
    \caption{Schematic illustration of gravitational interaction between test particles and a mechanical oscillator.}
    \label{fig2}
\end{figure}
In this section, we demonstrate that NMU can serve as a viable approach for certifying gravitationally induced entanglement between test particles of mass $m$ and a mechanical oscillator $M$, as described in Refs.~\cite{PRXQuantum.2.030330, PhysRevLett.128.110401}. As illustrated in Fig.~\ref{fig2}, we consider two test particles (labeled 1 and 2) positioned near the oscillator. Each particle can be prepared in a quantum superposition of two spatially distinguishable states, $\ket{\mathbf{x}_{1(2)\text{L}}}$ and $\ket{\mathbf{x}_{1(2)\text{R}}}$, separated by a distance $d_0 = |\mathbf{x}_{1(2)\text{L}} - \mathbf{x}_{1(2)\text{R}}|$. Here, L and R denote the left and right positions, respectively. The distance between the center of mass of the oscillator $\mathbf{x}_M$ and each test particle is given by $d = |\mathbf{x}_M - \mathbf{x}_1| = |\mathbf{x}_M - \mathbf{x}_2|$, with $\mathbf{x}_{1(2)} = (\mathbf{x}_{1(2)\text{L}} - \mathbf{x}_{1(2)\text{R}})/2$.

The interaction between the test particles and the oscillator is governed by the Newtonian gravitational potential, $V(\mathbf{x}_M, \mathbf{x}_m) = -{G M m}/{|\mathbf{x}_M - \mathbf{x}_m|}.$ Assuming $d \gg d_0, |\mathbf{x}_M|$, the total Hamiltonian can be written as
\begin{equation}
    \frac{H}{\hbar} = \tilde{\omega} a^\dagger a + \sum_{i=1,2} \omega_m \sigma_{z,i} + g \sigma_{z,i}(a + a^\dagger),
\end{equation}
where $a$ and $a^\dagger$ are the annihilation and creation operators of the mechanical oscillator, and the pseudo-spin operator is defined as
\begin{equation}
    \sigma_{z,i} = \ket{\mathbf{x}_{i,\text{L}}}\bra{\mathbf{x}_{i,\text{L}}} - \ket{\mathbf{x}_{i,\text{R}}}\bra{\mathbf{x}_{i,\text{R}}}.
\end{equation}
The oscillator's frequency is gravitationally shifted to
\begin{equation}
    \tilde{\omega} = \sqrt{\omega^2 - \frac{2Gm}{d^3}},
\end{equation}
where $\omega$ is the bare frequency. The frequency shift experienced by the test particles is given by
\begin{equation}
    \omega_m = \frac{GMm d_0}{2\hbar d^2},
\end{equation}
and the coupling strength between each particle and the oscillator is
\begin{equation}
    g = -\frac{G m d_0}{d^3} \sqrt{\frac{M}{2\hbar \tilde{\omega}}}. \label{eq: gravitation induced coupling}
\end{equation}

We consider both the test particles and the oscillator to be composed of silica with mass density $\varrho = 2400~\mathrm{kg\cdot m^{-3}}$. The radii of the test particles and the oscillator are $r = 70 \times 10^{-9}~\mathrm{m}$ and $R = 7 \times 10^{-6}~\mathrm{m}$, respectively. We set $d=175~\mathrm{\mu m}$, $d_0=500~\mathrm{nm}$. According to Eq.~\eqref{eq: gravitation induced coupling}, the coupling strength is inversely proportional to the oscillator frequency. Here, we consider a low-frequency oscillator with $\tilde{\omega} = 0.01~\mathrm{Hz}$, which is of the same order as that used in Ref.~\cite{PRXQuantum.2.030330}, resulting in a dimensionless coupling strength of $g/\tilde{\omega} \approx -2.75 \times 10^{-9}$. Also we consider that the oscillator is cooled to $T=1~\text{mK}$.

In this case, the maximal decay of the decoherence factors for this model is expected at the interaction time $t^*=\pi/\tilde{\omega}\approx 314~\text{s}$, yielding $|\phi_{ij,kl}(t=\pi/\tilde{\omega}) - \phi_{ij,kl}(0)|\sim 10^{-6}$. As reported in Ref.~\cite{PRXQuantum.2.030330}, such a long interaction time and small coherence change might be within reach of near-future experiments. Accessing the full-time dynamics remains highly demanding, making the NMU a more feasible tool for certifying particle-oscillator entanglement than HE incompatibility. Our numerical results show that $Q_\text{A}\approx 1.26 \times 10^{-6}$ at $t^*$, with the numerical precision $\epsilon\sim 10^{-10}$. For comparison, the baseline value obtained by randomly generated MU channel is $Q_{\text{A,MU}}\sim10^{-9}$. Thus, the estimated NMU exceeds the baseline by three orders of magnitude, providing strong evidence of gravitationally induced entanglement.

Finally, in comparison to the coherence revival detection scheme proposed in Ref.~\cite{PRXQuantum.2.030330}, a key advantage of the NMU approach is that it does not rely on the assumption that gravity-induced dephasing follows Lindblad dynamics with time-translational invariance. Therefore, the NMU can unambiguously certify entanglement generation, rather than ruling out certain classes of classical stochastic gravity models~\cite{PRXQuantum.2.030330,bassi2017gravitational,anastopoulos2022gravitational,PhysRevX.13.041040}.

\section{Discussion and Conclusion \label{sec5}}
We address the challenge of certifying system–environment entanglement in open quantum systems, focusing on the widely encountered class of pure dephasing dynamics. Motivated by the notion of HE incompatibility, we introduce a generalized and experimentally efficient approach based on the concept of NMU.

The NMU method overcomes key limitations of the HE incompatibility. Specifically, it is applicable to non-autonomous dynamics, enables entanglement certification at arbitrary time points without requiring access to the full-time evolution, and remains valid even when pure dephasing holds only approximately over a short-time regime. These features make the NMU approach an effective tool for entanglement detection.

We experimentally validate the method using a trapped-ion quantum processor, demonstrating that NMU can witness entanglement using fewer measurement resources than those required by the HE method. Finally, we propose that the NMU provides a promising strategy for certifying gravitationally induced entanglement in scenarios involving Newtonian interactions between massive particles, offering a conceptually transparent and experimentally feasible alternative to existing proposals.

Several open questions remain. The present work is limited to pure dephasing scenarios. It would therefore be valuable to explore the implications of system–environment correlations inferred from NMU in general unital dynamics. Furthermore, as suggested by the numerical results presented in Appendix~\ref{Appendix D}, we conjecture that NMU may serve not only as a witness of system–environment entanglement, but also as an indicator of its magnitude. Intuitively, the quantification of HE incompatibility proposed in Ref.~\cite{chen2019quantifying} may offer similar utility. We leave the formal validation of this conjecture as a promising direction for future work.

\section*{Acknowledgement}
The authors acknowledge fruitful discussions with Hong-Bin Chen. This work is supported by the National Center for Theoretical Sciences and National Science and Technology Council, Taiwan, Grants No. NSTC 114-2112-M-006-015-MY3. N. L. is supported by the RIKEN Incentive Research Program and by MEXT KAKENHI Grant Numbers JP24H00816, JP24H00820. A.M. was supported by the Polish National Science Centre (NCN) under the Maestro Grant No. DEC-2019/34/A/ST2/00081. F.N. is supported in part by: the Japan Science and Technology Agency (JST) [via the CREST Quantum Frontiers program Grant No. JPMJCR24I2, the Quantum Leap Flagship Program (Q-LEAP), and the Moonshot R\&D Grant Number JPMJMS2061], and the Office of Naval Research (ONR) Global (via Grant No. N62909-23-1-2074).

\appendix
\section{Formal proof of Theorem 1 \label{Appendix A}}
In the main text, we assume that $\rho_\e$ is full-rank and non-degenerate. Here, we relax this assumption such that a general environmental state can be expressed as
\begin{equation}
    \rho_\e = \sum_n q_n \Pi_n. \label{eq: general rhoe}
\end{equation}
Here, we set $q_n>0$ and $q_n\neq q_m~\forall~n\neq m$. Also, $\{\Pi_n\}$ forms an incomplete set of orthonormal projectors with $\sum_n \Pi_n < \mathbb{1}_\e$, implying that $\rho_\e$ is not full-rank. Therefore, we can express the Hilbert space of E as $\mathcal{H}_\e = \bigoplus_n \mathcal{H}_n \oplus \mathcal{H}_\perp$, where $\mathcal{H}_n$ denotes the eigensubspace defined by $\Pi_n$ and $\mathcal{H}_\perp$ represents the subspace that is orthogonal to the support of $\rho_\e$. In addition, we require $\text{rank}( \Pi_n)\geq 1$ to capture the possible degeneracy in $\rho_\e$. 

Let us now turn to the zero-discordant criteria described in Eq.~\eqref{eq: zero discord criteria}, which leads to 
\begin{equation}
    W_{ij} \Pi_n W_{ij}^\dagger = \Pi_n~\forall n.
\end{equation}
This implies that the unitary operator $W_{ij}$ preserves the eigen-subsapces and can be expressed as direct summation of local unitary operators acting on the eigen-spaces, namely, 
\begin{equation}
    W_{ij} = \bigoplus_n W_{ijn}\oplus W_{\perp,ij}~~\forall i,j,
\end{equation}
where $W_{ijn}$ and $W_{\perp,ij}$ are unitary operators respectively acting on the corresponding subspaces $\mathcal{H}_n$ and $\mathcal{H}_\perp$.
In general, the local unitary $W_{ijn}$ operator can be diagonalized under a basis for the corresponding eigen-subspace $\mathcal{H}_n$:
\begin{equation}
    W_{ijn} = \sum_{m} \exp({i\theta_{ijnm}}) \ket{e_{nm}}\bra{e_{nm}}.
\end{equation}
Similarly, we can also express the projectors with the same basis, i.e. $\Pi_n = \sum_m \ket{e_{nm}}\bra{e_{nm}}$. As mentioned in the main text, utilizing the relation $W_{ij}W_{ik}^\dagger = W_{kj}$, we obtain $W_{ijn}W_{ikn}^\dagger = W_{kjn}$ and $\theta_{ijnm}-\theta_{iknm} = \theta_{kjnm}$. This implies that the phase $\theta_{ijnm}$ can be written as a difference of two other phases, i.e., $\theta_{ijnm} = \phi_{inm}-\phi_{jnm} $. Therefore, the reduced evolution of S can be written as:
\begin{equation}
    \begin{aligned}
        \tilde{\rho}_\s &= \sum_{i,j} \rho_{\s,ij} \tr(W_{ij}\rho_\e) \ket{i}\bra{j} \\
        &= \sum_{n,m} q_n\exp(\phi_{inm}-\phi_{jnm})~\rho_{\s,i,j} \ket{i}\bra{j}  \\
        &= \sum_{n,m} q_n~ \tilde{U}_{nm}~\rho_\s~ \tilde{U}^\dagger_{nm},
    \end{aligned}
\end{equation}
with $\tilde{U}_{nm} = \exp(i\phi_{inm})\ket{i}\bra{i}$. This completes our proof of Theorem~\ref{thm1} for general environmental states.

\section{Corrected qutrit-like conditions \label{Appendix B}}

As shown in Ref.~\cite{PhysRevA.98.052344}, the qutrit-like conditions can be derived by treating the system as a qutrit, assuming the qubit-like conditions is already satisfied. Here, we follow a similar approach to obtain a corrected version of these conditions.

Consider now the system is prepared in a qutrit pure state:
\[\rho_\s =\ket{\psi}\bra{\psi} \text{~~with~~} \ket{\psi}= \sum_i c_i \ket{i}.\] 
In this case, the system-environment joint state $\rho_{\s\e}$ can be expressed by a $3\times 3$ block matrix, namely,
\begin{equation}
    \rho_{\s\e}=\begin{pmatrix}
        |c_0|^2~R & c_0 c_1^*~R_{01} & c_0 c_2^*~ R_{02}\\
        c_1 c_0^*~R_{10} & |c_1|^2~R & c_1 c_2^*~ R_{12}\\
        c_2 c_0^*~R_{20} & c_2 c_1^*~R_{21} & |c_2|^2~ R
    \end{pmatrix},
\end{equation}
Here, we define $R_{ij}=V_i \rho_\e V^\dagger_j~~\forall~i,j$. Due to the qubit-like conditions in Eq.~\eqref{eq: qubit-like condition}, we have $R_{ii}= R~~\forall i$.

Following a similar procedure in Eq.~\eqref{eq: general rhoe}, we consider the spectral decomposition of $R$, which can be written as
\begin{equation}
    R =\sum_n q_n \Pi_n.
\end{equation}
This implies that the Hilbert space of E can be decomposed as $\mathcal{H}_\e = \bigoplus_n \mathcal{H}_n \oplus \mathcal{H}_\perp$, where each $\mathcal{H}_n$ represents the eigensubspace corresponding to the projector $\Pi_n$, and $\mathcal{H}_\perp$ denotes the subspace orthogonal to the support of $R$.

The qubit-like conditions in Eq.~\eqref{eq: qubit-like condition} can also be expressed as $[R,Y_{ij}]=0$ with $Y_{ij}=V_i V_j^\dagger$. This also implies the unitary operator $Y_{i,j}$ can be decomposed as 
\begin{equation}
    Y_{ij} = \bigoplus_n Y_{ijn} \oplus Y_{\perp,ij},
\end{equation}
where $Y_{ijn}$ and $Y_{\perp,ij}$ are local unitary operators acting on the corresponding subspaces.
Therefore, we obtain
\begin{equation}
    R_{ij} = \bigoplus_n q_n Y_{ijn}\oplus\mathbf{0}_\perp.
\end{equation}
In other words, the local unitary operators $Y_{\perp,ij}$ becomes irreverent for determining whether $\rho_{\s\e}$ is separable. 

Since $[R,Y_{01}]=0$ implies $[\Pi_n, Y_{01n}]=0$, there exists a basis $\{\ket{e_{nm}}\}$ such that
\begin{equation}
\begin{aligned}
    \Pi_n =\sum_m \ket{e_{nm}}\bra{e_{nm}}
    \text{~and~} Y_{01n} =\sum_m e^{i\theta_{nm}} \ket{e_{nm}}\bra{e_{nm}}.
\end{aligned}
\end{equation}
Notably, $[R,Y_{ij}]=0$ does not guarantee $[Y_{01},Y_{02}]=0$, and therefore, $[Y_{01n},Y_{02n}]=0$ is not guaranteed. As a result, $Y_{02n}$ is in general not diagonalized under $\{\ket{e_{nm}}\}$. We can thus express it as
\begin{equation}
    Y_{02n} = \sum_{m,m'}x_{nmm'} \ket{e_{nm}}\bra{e_{nm'}}.
\end{equation}
Moreover, since $Y_{12}= V_1 V_2^\dagger= V_1 V_0^\dagger V_0 V_2^\dagger = Y_{01}^\dagger Y_{02}$, we have
\begin{equation}
    \begin{aligned}
    Y_{12n}&=Y_{01n}^\dagger Y_{02n}\\
    &=\sum_{m,m'} e^{-i\theta_{nm'}}~x_{nmm'}\ket{e_{nm}}\bra{e_{nm'}}.
    \end{aligned}
\end{equation}

To certify system-environment entanglement, we employ the Peres–Horodecki criterion~\cite{PhysRevLett.77.1413,
HORODECKI19961}, which involves performing a partial transposition on the system:
\begin{equation}
    \rho_{\s\e}^{T_\s}=\begin{pmatrix}
        |c_0|^2~R & c_1 c_0^*~R_{10} & c_2 c_0^*~ R_{20}\\
        c_0c_1^*~R_{01} & |c_1|^2~R & c_2 c_1^*~ R_{21}\\
        c_0c_2^*~R_{02} & c_1 c_2^*~R_{12} & |c_2|^2~ R
    \end{pmatrix}.
\end{equation}
It is known that $\rho_{\s\e}^{T_\s}$ must be positive semidefinite if $\rho_{\s\e}$ is separable. As shown in Ref.~\cite{PhysRevA.98.052344}, this condition can be verified by checking a class of principal minors defined as 
\begin{widetext}
    \begin{equation}
    \begin{aligned}
        D_{(nm);(n'm')} &= \text{det}
        \begin{pmatrix}
            |c_0|^2~[R]_{(nm);(nm)} & c_1 c_0^*~[R_{1,0}]_{(nm);(nm)} & c_2 c_0^*~ [R_{2,0}]_{(n'm');(nm)}\\
        c_0c_1^*~[R_{0,1}]_{(nm);(nm)} & |c_1|^2~[R]_{(nm);(nm)} & c_2 c_1^*~ [R_{2,1}]_{(n'm');(nm)}\\
        c_0c_2^*~[R_{0,2}]_{(nm);(n'm')} & c_1 c_2^*~[R_{1,2}]_{(nm);(n'm')} & |c_2|^2~[R_{2,2}]_{(n'm');(n'm')}
        \end{pmatrix}\\
        &=\text{det}
        \begin{pmatrix}
            |c_0|^2~q_n & c_1 c_0^*~q_n e^{-i\theta_{nm}} & c_2 c_0^*~ q_n x^*_{nmm'}\delta_{nn'}\\
        c_0c_1^*~q_n e^{i\theta_{nm}} & |c_1|^2~q_n & c_2 c_1^*~ q_n e^{i\theta_{nm'}} x^*_{nmm'}\delta_{nn'}\\
        c_0c_2^*~q_n x_{nmm'}\delta_{nn'} & c_1 c_2^*~q_n e^{-i\theta_{nm'}} x_{nmm'}\delta_{nn'} & |c_2|^2~q_{n'}
        \end{pmatrix},
    \end{aligned}
    \end{equation}
\end{widetext}
with $[R_{ij}]_{(nm);(n'm')} = \bra{nm}R_{i,j}\ket{n'm'}$. If any of these principal minors is negative, it implies that $\rho_{\s\e}^{T_\s}$ is not positive semidefinite, and thus the system-environment state is entangled.

For the cases where $n\neq n'$, we find that $D_{(nm);(n'\neq n~m')} =0$.
This means that it suffices to consider the cases where $n'=n$, for which the principal minors can be written as 
    \begin{equation}
    \begin{aligned}
        &D_{(nm);(nm')}\\
    &=-2|c_0 c_1 c_2|^2 q_n^3 |x_{nmm'}|^2\left[1-\cos(\theta_{nm}-\theta_{nm'}) \right].
    \end{aligned}
    \end{equation}
From this, we see that $D_{(nm);(nm')}=0$ only if $\cos(\theta_{nm}-\theta_{nm'})=1$ or $x_{nmm'}=0$. Otherwise, entanglement can be certified by a negative principal minor. Furthermore, if all of these principal minors vanish, it implies that 
\begin{equation}
    [Y_{ijn},Y_{kln}]=0~~~~\forall i,j,k,l.\label{eq: corrected qutrit-like conditions 1}
\end{equation}
In this case, we can also conclude that $[R_{ij},R_{kl}]=0$, which indicates that $\rho_{\s\e}$ is a separable state and zero-discordant state from E to S. Therefore, the corrected form of the qutrit-like conditions should be given by Eq.~\eqref{eq: corrected qutrit-like conditions 1}, or equivalently,   
\begin{equation}
    R[V_iV_j^\dagger,V_kV_l^\dagger]R=0.~\label{eq: corrected qutrit-like conditions 2}
\end{equation}

\section{Formal proof of Theorem 2 with corrected qutrit-like criteria \label{Appendix C}}
In the main text, we show that the separability criteria in Eqs.~\eqref{eq: qubit-like condition} and \eqref{eq: qutrit-like condition} implies the zero-discordant condition in Eq.~\eqref{eq: zero discord criteria}. However, the converse does not hold, which leads to an apparent inconsistency.

Here, we show that the inconsistency can be resolved by considering the corrected qutrit-like conditions in Eq.~\eqref{eq: corrected qutrit-like conditions 2}. First, as stated in Eq.~\eqref{eq: zero discord criteria 2} of the main text, the zero-discord criteria imply $V_i \rho_\e^2 V_i^\dagger=V_j \rho_\e^2 V_j^\dagger$, which is equivalent to $V_i \rho_\e V_i^\dagger=V_j \rho_\e V_j^\dagger=R$. This means that the zero-discord conditions leads to the qubit-like conditions. Moreover, through a straightforward algebraic procedure, we can demonstrate that
\begin{equation}
    \begin{aligned}
    &[V_i \rho_\e V_j^\dagger, V_k \rho_\e V_l^\dagger] =0 \\
    \implies& V_i \rho_\e V_j^\dagger V_k \rho_\e V_l^\dagger = V_k\rho_\e V_l^\dagger V_i \rho_\e V_j^\dagger \\
    \implies &V_i \rho_\e (V_i^\dagger V_i) V_j^\dagger V_k (V_l^\dagger V_l)\rho_\e V_l^\dagger \\&= V_k\rho_\e (V_k^\dagger V_k) V_l^\dagger V_i (V_j^\dagger V_j) \rho_\e V_j^\dagger \\
    \implies & R V_i V_j^\dagger V_k V_l^\dagger R = R V_k V_l^\dagger V_i V_j^\dagger R\\
    \implies &R[V_i V_j^\dagger , V_k V_l^\dagger ]R =0.
    \end{aligned}
\end{equation}
Therefore, we conclude that the zero-discord conditions also imply the corrected qutrit-like conditions, completing the proof.

\section{Conjecture: estimating negativity by non-mixed unitarity \label{Appendix D}}
\begin{figure}
    \centering
    \includegraphics[width=0.9\linewidth]{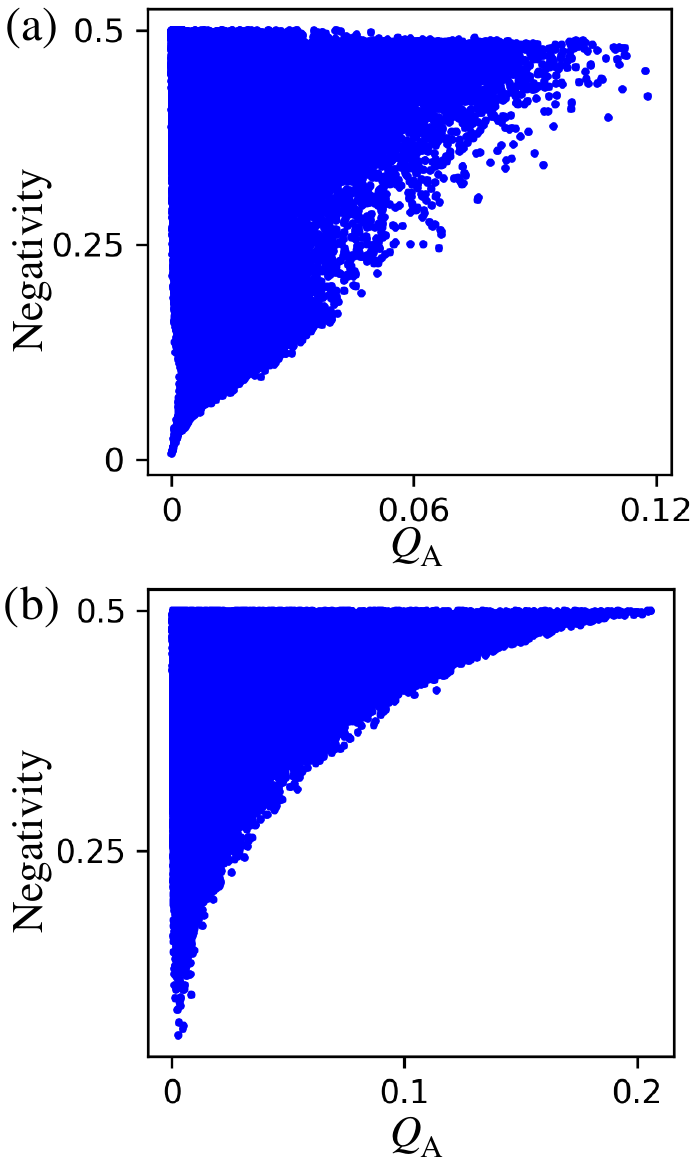}
    \caption{Negativity of the system-environment state [Eq.~\eqref{eq: negativity}] verses the NMU measure $Q_\text{A}$ for (a) $d_\s = 3$ and (b) $d_\s=4$ with randomly sampled pure dephasing maps.  }
    \label{fig3}
\end{figure}

Here, we randomly generate pure dephasing maps for three-level ($d_\s = 3$) and four-level ($d_\s =4$) systems, each coupled to a qubit environment initially prepared in a pure state, and investigate the relationship between the NMU and the resulting system–environment entanglement.

In each run, one can consider a qubit initial state $\ket{\psi}$ and a set of local unitaries acting on the environment qubit $\{V_i\}_{i=0,\cdots,d_\s-1}$ to determine the system-environment global unitary described by [see also Eq.~\eqref{eq: U_SE}]:
\begin{equation}
    U_{\s\e} =\sum_{i=0}^{d_\s-1}\ket{i}\bra{i}\otimes V_i. \label{eq: U_SE appendix}
\end{equation}
The corresponding pure dephasing map is completely specified by $\{V_i\ket{\psi} \}_{i=0,\cdots,d_\s-1}$. To increase the sampling efficiency, rather than separately generating the set of local unitaries and the initial qubit state, we instead directly generate a set of random pure qubit states, $\{\ket{\psi_i}\}_{i=0,\cdots,d_\s}$. We then compute the NMU quantified by $Q_\text{A}$ [defined by Eq.~\eqref{eq: QA}] for the assicated pure dephasing map.

As shown in Ref.~\cite{takahashi2025coherence}, the amount of entanglement generated in such a pure dephasing scenario is upper bounded by the initial coherence of the system. To maximize this potential entanglement, we initialize the system in the maximally coherent state $\ket{\psi_\s}=\sum_{i=0}^{d_\s-1}\ket{i}/\sqrt{d_\s}$. In this case the joint system-environment state evolves into
\begin{equation}
    \ket{\Psi_{\s\e}} =\frac{1}{\sqrt{d_\s}} \sum_{i=0}^{d_\s-1}\ket{i} \otimes \ket{\psi_i} 
    \label{eq: radom Psi_SE}
\end{equation}
To quantify the generated entanglement, we compute the negativity of the joint state, which is given by
\begin{equation}
    \mathcal{N}=\frac{  \|\rho_{\s\e}^{T_\s}\|_1 -1 }{2}, \label{eq: negativity}
\end{equation}
where $\rho_{\s\e}=\ket{\Psi_{\s\e}}\bra{\Psi_{\s\e}}$, $T_\s$ denotes the partial transposition on S, and $\|\cdot\|_1$ is the trace norm.

In Fig.~\ref{fig3}, we present results based on $10^{6}$ and $10^{5}$ randomly sampled pure dephasing maps for $d_\s=3$ and $d_\s =4$, respectively. For each sample, we compute both NMU and the resulting entanglement. For both cases, the results show a monotonic increase in the minimum negativity as $Q_\text{A}$ grows, indicating that NMU may not only serve as witness but also a quantifier of system-environment entanglement under pure dephasing scenario.

\section{Circuit decomposition of the controlled-controlled rotation gate \label{Appendix E}}
\begin{figure}
    \centering
    \includegraphics[width=1\linewidth]{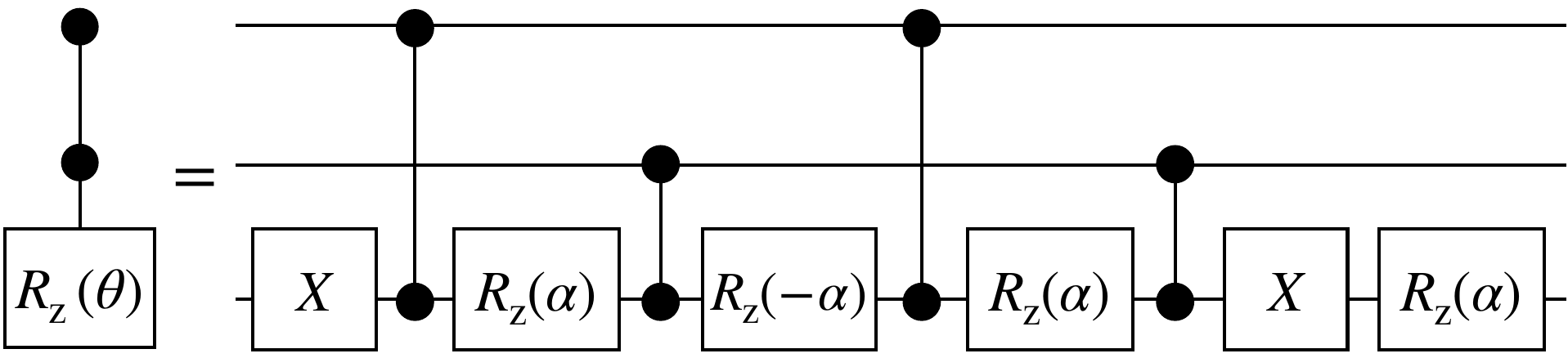}
    \caption{Circuit decomposition of a CCRZ gate with $\alpha = \theta/4$.}
    \label{fig4}
\end{figure}
Here, we present circuit decomposition of controlled-controlled rotation gate into CZ gate and single-qubit gates. As shown in Fig.~\ref{fig4}, a controlled-controlled rotation Z (CCRZ) gate can be decomposed into four controlled Z (CZ) gates (represented by a vertical line with solid circles at the ends), two X gates, and four rotation-Z ($\text{R}_{\text{z}}$) gates, where the gates are defined as 
\begin{equation}
    \begin{aligned}
         \text{CCRZ}(\theta) &= \sum_{\substack{i,j \in \{0,1\} \\ (i,j) \neq (1,1)}}\ket{i,j}\bra{i,j}\otimes \mathbb{1} + \ket{1,1}\bra{1,1}\otimes \text{R}_\text{z}(\theta),\\
        \text{CZ} &= \ket{0}\bra{0}\otimes \mathbb{1} + \ket{1}\bra{1}\otimes \sigma_z, \\
        \text{R}_\text{z}(\theta) &= \cos(\theta/2)\mathbb{1}-i\sin(\theta/2)\sigma_z,\\
        X  &= \ket{0}\bra{1} + \ket{1}\bra{0}.
    \end{aligned}
\end{equation}

Based on the decomposition of the CCRZ gate, this approach can be directly generalized to a controlled-controlled rotation gate with an arbitrary rotation axis. Specifically, by applying appropriate single-qubit gates to the target qubit, one can modify the rotation axis as desired. Furthermore, by applying X gates to the control qubits, one can alter the control conditions, enabling the realization of the white control circles shown in Fig.~\ref{fig1}(a) of the main text.

\bibliography{ref}

\end{document}